\documentclass[11pt]{article}

\usepackage[T1]{fontenc}
\usepackage[utf8]{inputenc}

\usepackage[margin = 0.9in]{geometry}
\usepackage{booktabs} 
\usepackage{enumitem}

\usepackage{latexsym,amsmath,amsfonts,amssymb,stmaryrd,mathtools}
\usepackage{amsthm}
\usepackage{xcolor}
\usepackage{algorithm, algorithmicx}
\usepackage[noend]{algpseudocode}
\usepackage{graphicx}
\usepackage{hyperref}
\usepackage{wrapfig}
\usepackage{listings, fancyvrb, multicol}
\usepackage{multirow}
\usepackage{times}
\usepackage{comment}
\usepackage{authblk}
\usepackage{xspace}
\usepackage{pifont}
\usepackage{parcolumns}

\usepackage{makecell}

\definecolor{ao}{rgb}{0.0, 0.5, 0.0}

\newtheorem{theorem}{Theorem}[section]
\newtheorem{lemma}[theorem]{Lemma}

\newtheorem{claim}[theorem]{Claim}

\newtheorem{definition}[theorem]{Definition}
\newtheorem{observation}[theorem]{Observation}

\newcommand{\cfont}[1]{\mbox{\tt\bf\small #1}}

\newcommand{\hl}[1]{\colorbox{cyan!30}{#1}}

\newcommand{\var}[1]{\cfont{#1}}

\newcommand{\rd}{\cfont{read}}
\newcommand{\wt}{\cfont{write}}

\newcommand{\retire}{\cfont{retire}}
\newcommand{\swcopy}{\cfont{swcopy}}
\newcommand{\mwcopy}{\cfont{copy}}
\newcommand{\memcopy}{\cfont{memory-to-memory move}}
\newcommand{\wllsc}{\cfont{WeakLLSC}}
\newcommand{\wll}[1]{\cfont{wLL}}
\newcommand{\wvl}[1]{\cfont{VL}}
\newcommand{\wsc}[1]{\cfont{SC}}
\newcommand{\LL}[1]{\cfont{LL}}
\newcommand{\VL}[1]{\cfont{VL}}
\newcommand{\SC}[1]{\cfont{SC}}
\newcommand{\CL}{\cfont{CL}}
\newcommand{\dest}{\cfont{Destination}}

\newcounter{results}
\newcommand{\result}[2]{\refstepcounter{results} \vspace{.03in}{\flushleft\textbf{Result \theresults~(#1)}: \emph{#2}}
\vspace{.05in}}

\newcommand{\hide}[1]{}

\makeatletter
\lst@Key{countblanklines}{true}[t]%
{\lstKV@SetIf{#1}\lst@ifcountblanklines}

\def\StartLineAt#1{\lstset{firstnumber=#1}}

\lst@AddToHook{OnEmptyLine}{%
  \lst@ifnumberblanklines\else%
  \lst@ifcountblanklines\else%
  \advance\c@lstnumber-\@ne\relax%
  \fi%
  \fi}
\makeatother

\usepackage{subcaption} 

\begin{document}

	\title{LL/SC and Atomic Copy: Constant Time, Space Efficient Implementations using only pointer-width CAS}
	
	  \author[1]{Guy E. Blelloch}
	  \author[1]{Yuanhao Wei}
	  \affil[1]{Carnegie Mellon University, USA}
	  \affil[ ]{\textit{\{guyb, yuanhao1\}@cs.cmu.edu}}

	\maketitle	
	\begin{abstract}
When designing concurrent algorithms, Load-Link/Store-Conditional (LL/SC) is 
often the ideal primitive to have because unlike Compare and Swap (CAS), LL/SC is immune to the ABA problem.
Unfortunately, the full semantics of LL/SC are not supported in hardware by any modern 
architecture, so there has been a significant amount of work on simulations of 
LL/SC using CAS, a synchronization primitive that enjoys 
widespread hardware support. However, all of the algorithms so far 
that are constant time either use unbounded sequence numbers (and thus 
base objects of unbounded size), or require $\Omega(MP)$ space for $M$ LL/SC object (where $P$ is the number of processes).

We present a constant time implementation of $M$ LL/SC objects using $\Theta(M+kP^2)$ space (where $k$ is the number of outstanding LL operations per process) and requiring only pointer-sized CAS objects. 
In particular, using pointer-sized CAS objects means that we do not use unbounded sequence numbers.
For most algorithms that use LL/SC, $k$ is a small constant, so this result implies that these algorithms can be implemented directly from CAS objects of the same size with asymptotically no time overhead and $\Theta(P^2)$ additive space overhead.
This $\Theta(P^2)$ extra space is paid for once across all the algorithms running on a system, so asymptotically, there is little benefit to using CAS over LL/SC.
Our algorithm can also be extended to implement $L$-word $LL/SC$ objects 
in $\Theta(L)$ time for $LL$ and $SC$, $O(1)$ time for $VL$, and $\Theta((M+kP^2)L)$ space.

To achieve these bounds, we begin by implementing a new primitive called
Single-Writer Copy which takes a pointer to a word sized memory location 
and atomically copies its contents into another object. The 
restriction is that only one process is allowed to write/copy into the destination object at a time.
We believe this primitive will be very useful in designing other concurrent algorithms as well. 
\end{abstract}

\section{Introduction}

In lock-free, shared memory programming, it's well known that the choice 
of atomic primitives makes a big difference in terms of ease of 
programability, efficiency, and even computability. 
Most processors today support a set of basic synchronization primitives such 
as Compare-and-Swap, Fetch-and-Add, Fetch-and-Store, etc. However, many 
useful primitives are not supported, which motivates the need for efficient 
software implementations of these primitives. In this work, we present 
constant time, space-efficient implementations of a widely used primitive 
called Load-Link/Store-Conditional (LL/SC) as well as a new primitive we 
call Single-Writer Copy (\swcopy{}). All our implementations use only 
pointer-width read, write, and CAS. In particular, restricting ourselves 
to pointer-width operations means that we do not use unbounded sequence 
numbers, which are often used in other LL/SC from CAS implementations 
\cite{moir1997practical, michael2004practical, jayanti2005many}. 
Many other algorithms based on CAS also use unbounded sequence numbers 
(often alongside double-wide CAS) to get around the ABA problem and this is sometimes called the \emph{IBM tag methodology}\cite{michael2004aba, ibmtag}. 
Our LL/SC implementation can be used to avoid the use of 
unbounded sequence numbers and double-wide CAS in these algorithms.


We implemented a Single-Writer Atomic Copy (\swcopy{}) primitive and found
that it greatly simplified our implementation of LL/SC. We believe it will be 
useful in a wide range of other applications as well. The \swcopy{} 
primitive can be used to atomically read one memory location and write the 
result into another. The memory location being read can be arbitrary, 
but the location being written to has to be a special \dest{} object. A 
\dest{} object supports three operation, \rd{}, \wt{}, and 
\swcopy{} and it allows any process to \rd{} from it, but only a 
single process to \wt{} or \swcopy{} into it. 
We expect this primitive to be very useful in concurrent algorithms that use 
announcement arrays as it allows the algorithm to atomically read a 
memory location and announce the value that it read.
This primitive can be used to solve various problems related
to resource management, such as concurrent reference
counting, in a constant (expected) time, wait-free manner~\cite{blelloch2020concurrent}.

In this work, we focus on wait-free solutions. Roughly speaking, 
\emph{wait-freedom} ensures that \emph{all} processes are making progress 
regardless of how they are scheduled. 
In particular, this means wait-free algorithms do not suffer from problems such 
as deadlock and livelock.
All algorithms in this paper take in $O(1)$ or $O(L)$ time (where $L$ is the number of words spanned by the implemented object), which is stronger than wait-freedom. The correctness condition we consider is \emph{linearizability}, which intuitively means that all operations appear to take effect at a single point.


In our results below, the time complexity of an operation is the number of instructions that 
it executes (both local and shared) in a worst-case execution and space complexity of an object is 
the number of words that it uses (both local and shared). Counting local objects/operations is consistent with previous papers on the topic~\cite{anderson1995large, michael2004practical, jayanti2005many}.
There as been a significant amount of prior work on implementing LL/SC from CAS~\cite{anderson1995multi, moir1997practical, jayanti2003efficient, jayanti2005many, michael2004practical} and we discuss them in more detail in Section \ref{sec:related}.

\result{Load-Link/Store-Conditional} {
\label{result:llsc}
  A collection of $M$ LL/SC objects operating on $L$-word values shared by $P$ processes, each performing at most $k$ outstanding LL operations, can be implemented with:
  \begin{enumerate}
    \item $\Theta(L)$ time for LL and SC, $O(1)$ time for VL,
    \item $\Theta((M+kP^2)L)$ space,
    \item single word (at least pointer-width) read, write, CAS.
  \end{enumerate}
}


For many data structures implemented from LL/SC, such as Fetch-And-Increment~\cite{ellen2013optimal} and various Universal Construction~\cite{herlihy1993methodology, barnes1993method, afek1995wait}, $k$ is at most 2.
Our result implies that we can implement any number of such data structures from \emph{equal sized} CAS while maintaining the same time complexities and using only $\Theta(P^2)$ additional space across all the objects.
In contrast, using previous approaches~\cite{anderson1995large, anderson1995multi, jayanti2003efficient, moir1997practical} to implementing LL/SC from equal sized CAS would require $\Omega(P)$ space \emph{per} LL/SC object.
$\Theta(P^2)$ space overhead is very small compared to the memory size of most machines, so this result says that there is almost no disadvantage, from an asymptotic complexity perspective, to using LL/SC rather than CAS.



We also implement a \dest{} object supporting \rd{}, \wt{} and \swcopy{} with the following bounds.

\result{Single-Writer Copy} {
\label{result:swcopy}
  A collection of $M$ \dest{} objects shared by $P$ processes can be implemented with:
  \begin{enumerate}
    \item $O(1)$ worst-case time for read, write, and \swcopy{}
    \item $\Theta(M+P^2)$ space
    \item single word (at least pointer-width) read, write, CAS.
  \end{enumerate}
}

To help implement the \dest{} objects, we implement a weaker version of LL/SC with the bounds below. Our version of weak LL/SC is a little different from what was previously studied~\cite{anderson1995large, jayanti2003efficient, moir1997practical}. We compare the two in more detail in Section \ref{sec:related}.

\result{Weak Load-Link/Store Conditional} {
\label{result:wllsc}
 A collection of $M$ weak LL/SC objects operating on $L$-word values shared by $P$ processes, each performing at most one outstanding wLL, can be implemented with:
  \begin{enumerate}
    \item $\Theta(L)$ time for wLL and SC, $O(1)$ time for VL,
    \item $\Theta((M+P^2)L)$ space,
    \item single word (at least pointer-width) read, write, CAS.
  \end{enumerate}   
}

Our implementations of \swcopy{} and LL/SC are closely related. 
We begin in Section \ref{sec:wllsc} by 
implementing a weaker version of LL/SC (Result \ref{result:wllsc}). 
Then, in Section \ref{sec:swcopy}, we use this weaker LL/SC to implement 
\swcopy{} (Result \ref{result:swcopy}), and finally, in Section \ref{sec:llsc}, 
we use \swcopy{} to implement the full semantics of LL/SC (Result \ref{result:llsc}). As we shall 
see, once we have \swcopy{}, our algorithm for regular LL/SC is 
almost the same as our algorithm for weak LL/SC.

\section{Related Work}
\label{sec:related}

\textbf{LL/SC from CAS.} Results for implementing LL/SC from CAS are summarized in Table \ref{tab:llsc-results}. The column titled ``Size of LL/SC Object'' lists the largest possible LL/SC object supported by each algorithm. For example, $W-2\log{}P$ means that the implemented LL/SC object can store at most $W-2\log{}P$ bits, and $LW$ means that the implemented object can be arbitrarily large. All the algorithm shown in the table are wait-free and have optimal time bounds.
The time and space bounds listed in Table \ref{tab:llsc-results} are for the common case where $k$ is a constant.

So far, all previous algorithms suffer from one of three drawbacks. They either (1) are not wait-free constant time~\cite{doherty2004bringing, israeli1994disjoint}, (2) use unbounded sequence numbers~\cite{moir1997practical, michael2004practical, jayanti2005many, jayanti2005unknown}, or (3) require $\Omega(MP)$ space~\cite{anderson1995multi, jayanti2003efficient, aghazadeh2016tag, anderson1995large, jayanti2005large, moir1997practical}. There are also some other desirable properties that an algorithm can satisfy. For example, the algorithms by Jayanti and Petrovic~\cite{jayanti2005unknown} and Doherty et al.~\cite{doherty2004bringing} do not require knowing the number of processes in the system. Also, some algorithms are capable of implementing multi-word LL/SC from single-word CAS, whereas others only work when LL/SC values are smaller than word size. 



\textbf{Weak LL/SC from CAS.} A variant of \wllsc{} was introduced by~\cite{anderson1995large} and also studied in~\cite{jayanti2003efficient, moir1997practical}. The version we consider is even less restrictive than theirs because they require a failed \wll{} operation to return the process id of the \wsc{} operation that caused it to fail whereas we don't require failed \wll{} operations to return anything. While prior work is able to implement the stronger version of \wll{}, they either employ stronger primitives like LL/SC~\cite{anderson1995large}, use unbounded sequence numbers~\cite{moir1997practical}, require $O(MP)$ space for $M$ \wllsc{} objects~\cite{anderson1995large, jayanti2003efficient}, or require storing $(4\log{} P)$-bits in a single word~\cite{jayanti2003efficient}. To match the bounds stated in Result \ref{result:wllsc}, we define and implement a version of weak LL/SC that is sufficient for our \swcopy{} algorithm. Conveniently, the majority of our weak LL/SC algorithm from Section \ref{sec:wllsc} can be reused when implementing full LL/SC in Section \ref{sec:llsc}.

\textbf{Atomic Copy.} A similar primitive called \memcopy{} was studied in Herlihy's wait-free hierarchy paper~\cite{herlihy1991wait}. The primitive allows atomic reads and writes to any memory location and supports a \var{move} instruction which atomically copies the value at one memory location into another. Herlihy showed that this primitive has consensus number infinity. Our \swcopy{} is a little different because it allows arbitrary atomic operations (e.g. Fetch-and-Add, Compare-and-Swap, Write, etc) on the source memory location as long as the source objects supports an atomic read. Another difference is that we restrict the destination of the copy to be single-writer. Herlihy's proof that \memcopy{} has unbounded consensus number would also work with our \swcopy{} primitive. This means that \swcopy{} objects (or more precisely, the \dest{} objects defined in Section \ref{sec:swcopy-alg}) also have consensus number infinity.

{\renewcommand{\arraystretch}{1.5}
\begin{table}
\small
  \centering
    \begin{tabular}{ | l | c | c | c | c | }
    \hline
    \textbf{Prior Work} & \thead{\textbf{Word Size (W)}} & \thead{\textbf{Size of} \\ \textbf{LL/SC Object}} & \textbf{Time} & \textbf{Space}  \\
    \hline
    \multicolumn{1}{|m{3cm}|}{Anderson and Moir~\cite{anderson1995multi}, Figure 1} & $W > 2\log P$ & $W - 2\log P$ & $O(1)$ & $O(P^2M)$ \\
    \hline
    Moir~\cite{moir1997practical}, Figure 4 & unbounded\footnotemark[1] & $W -$ tag\_size & $O(1)$ & $O(P + M)$  \\
    \hline
    Moir~\cite{moir1997practical}, Figure 7 & $W > 3\log P$ & $W - 3\log P$  & $O(1)$ & $O(P^2 + PM)$  \\
    \hline
    Jayanti and Petrovic~\cite{jayanti2003efficient} & $W \geq 4\log P$ & $W$ & $O(1)$ & $O(PM)$  \\
    \hline
    Michael~\cite{michael2004practical}  & unbounded\footnotemark[1] & $LW$ & $O(L)$\footnotemark[2] & $O((P^2 + M)L)$  \\
    \hline
    Jayanti and Petrovic~\cite{jayanti2005many} & unbounded\footnotemark[1] & $LW$ & $O(L)$ & $O((P^2 + M)L)$  \\
    \hline
    Jayanti and Petrovic~\cite{jayanti2005unknown} & unbounded\footnotemark[1] & $W$ & $O(1)$ & $O(P^2 + PM)$ \\
    \hline
    Aghazadeh et al.~\cite{aghazadeh2016tag} & $W \geq 2\log M + 6\log P$ & $LW$ & $O(L)$ & $O(MP^5L)$ \\
    \hline
    \multicolumn{1}{|m{3cm}|}{Anderson and Moir~\cite{anderson1995large}, Figure 2} & $W \geq$ ptr\_size & $LW$ & $O(L)$ & $O(P^2ML)$  \\
    \hline
    \textbf{This Paper} & $W \geq$ ptr\_size & $LW$ & $O(L)$ & $O((P^2 + M)L)$  \\
    \hline
    \end{tabular}
    \caption{Cost of implementing $M$ LL/SC variables from
      CAS. Size is measured in number of bits. 
      The time and space bounds are presented in the common case where the maximum number of outstanding LL operations is a constant. }
    \label{tab:llsc-results}
\end{table}}

\footnotetext[1]{Uses unbounded sequence numbers}
\footnotetext[2]{Amortized expected time}

\section{Preliminaries}

We work in the standard asynchronous shared memory model~\cite{attiya2004distributed} with $P$ processes communicating through base objects that are either registers, CAS objects, or LL/SC objects. Processes may fail by crashing.
All base objects are word-sized and we assume they are large enough to store pointers into memory.

In our model, an \emph{execution} (or equivalently, \emph{execution history}) is an alternating sequence of \emph{configurations} and \emph{steps} $C_0$, $e_1$, $C_1$, $e_2$, $C_2$, $\dots$, where $C_0$ is an \emph{initial configuration}. Each step is a shared operation on a base object. Configuration $C_i$ consists of the state of all base objects, and every process after the step $e_i$ is applied to configuration $C_{i-1}$.

If configuration $C$ proceeds configuration $C'$ in an execution, the \emph{execution interval} from $C$ to $C'$ is the set of all configurations and steps between $C$ and $C'$, inclusive. Similarly, the \emph{execution interval} of an operation is the set of all configurations and steps from the first step of that operation to the last step of that operation. The execution interval for an \emph{incomplete operation} is the set of all configurations and steps starting from the first step of that operation. 

We say the implementation of an object is \emph{linearizable}~\cite{herlihy1990linearizability} if, for
every possible execution and for each operation on that object in the execution, we can pick a configuration or step in its execution interval to be its linearization point, such that the operation appears to occur instantaneously at this point. In other words, all operations on the object must behave as if they were performed sequentially, ordered by their linearization points. If multiple operations have the same linearization point, then an ordering must be defined among these operations.

All implementations that we discuss will be \emph{wait-free}. This means that each operation by any non-faulty process $p_i$ is guaranteed to complete within a finite number of steps by $p_i$. 

Consider an execution where the base objects are LL/SC objects.
If a process performs an LL operation, then the LL is considered to be \emph{outstanding} until the process performs a corresponding SC on the same object.
For algorithms that use LL/SC as base objects, we frequently use $k$ to denote the maximum number of outstanding LL operations per process at any point during an execution.
If $k = 1$, then each process alternates between performing LL and SC.

	
\section{Weak LL/SC from CAS}
\label{sec:wllsc}

As a subroutine, our \swcopy{} operation makes use of a weaker version of 
LL/SC. This weaker version supports three operations \wll{}, \wvl{} and \wsc{}, and works 
the same way as regular LL/SC except that \wll{} is allowed to not return anything 
if the subsequent \wsc{} is guaranteed to fail. We call a \wll{} operation 
\emph{successful} if it returns a value. Otherwise, we call it \emph{unsuccessful}. 
We call a \wsc{} operation successful if it returns true and unsuccessful otherwise.
Note that a \wll{} operation can only be unsuccessful if it is concurrent with 
a successful \wsc{}. We assume that \wvl{} and \wsc{} are only performed if the previous \wll{} by that process was successful.



In Section \ref{sec:wllsc-alg}, we present a constant time algorithm for 
weak LL/SC in the case where the maximum number of outstanding \wll{} operations per process is one.
This version of weak LL/SC is sufficient to implement the other algorithms in our paper.
An \wll{} operation is considered \emph{outstanding} if it is successful and there has not yet been a corresponding \wsc{} operation.

\subsection{Implementation of Weak LL/SC}
\label{sec:wllsc-alg}

In this section, we show how to implement $M$ \wllsc{}
objects, each spanning $L$-words, in wait-free constant time and 
$O((M+P^2)L)$ space. 
The high level idea is to use a layer 
of indirection 
and use an algorithm similar to Hazard Pointers~\cite{michael2004hazard} to upper bound the memory usage. 
Each \wllsc{} object is represented using  
a pointer, \var{buf}, to an $L$-word buffer storing the current value of the object. 
To perform an \wsc{}, the process first allocates a new $L$-word buffer, writes the new value in it, and then tries to write a pointer to this buffer into \var{buf} 
with a CAS. A \wll{} operation simply reads \var{buf} and returns the
value that it points to. The problem with this algorithm
is that it uses an unbounded amount of space. Our goal is to recycle 
buffer objects so that we use at most $O(M+P^2)$ of them. The idea of recycling 
buffers is an important part of many previous algorithms 
\cite{jayanti2005many, michael2004aba, aghazadeh2014making}. 
However, since we are only interested in implementing 
\wllsc{}, we are able to avoid using unbounded sequence numbers and provide 
better time/space complexities.

We recycle buffers with a variant of Hazard Pointers that is \emph{worst-case} constant time rather than \emph{expected} constant time.
Before accessing a buffer, a \wll{} operation has to first protect it by 
writing its address to an announcement array. To make sure that its announcement
happened ``in time'', the \wll{} operation re-reads \var{buf} and makes sure it is the same as what was announced. If \var{buf} has changed, then the \wll{} operation can return 
\var{empty} because it must have been concurrent with a successful \wsc{} and it can
linearize immediately before the linearization point of the \wsc{}. If \var{buf} is equal to the announced
pointer, then the buffer has been protected and the \wll{} operation can 
safely read from it.

A \wvl{} operation by process $p_i$ simply checks if \var{buf} is equal to the buffer announced by its previous \wll{} operation. 
If so, it returns true, otherwise, it returns false.

For the purpose of the \wsc{} operation, each process maintains two lists 
of buffers: a free list (\var{flist}) and a retired list (\var{rlist}).
In a \wsc{} operation, the process allocates by popping a buffer off its 
local free list. If the CAS instruction performed by the \wsc{} is successful, it adds the old value of 
the CAS to its retired list. Each process's free list starts off with $2P$ 
buffers and we maintain the invariant that the free list and retired list 
always add up to $2P$ buffers. When the free list becomes empty and the 
retired list hits $2P$ buffers, the process moves some buffers from the 
retired list to the free list. To decide which buffers are safe to reuse, 
the process scans the announcement array (the scan doesn't have to be atomic) and 
moves a buffer from the retired list to the free list if it was not seen 
in the array. Since the process sees at most $P$ different buffers in
the announcement array during its scan, its free list's size is guaranteed to be at 
least $P$
after this step. 
In a later paragraph, we show how this step can be 
performed in worst-case $O(P)$ time, which amortizes over the 
number of free buffers found. 

Pseudo-code is shown in Figure \ref{alg:weakllsc}. In the pseudo-code, we 
use \var{A[i].}\rd{} and 
\var{A[i].}\wt{} to read from and write to the announcement array \var{A}. 
Since each element of the announcement array is a pointer type,
\rd{} and \wt{} are
trivially implemented using the corresponding atomic instruction. 
We wrap these instructions in function calls so that the code can be 
reused in Section \ref{sec:llsc-alg}.
The argument from the previous paragraph also implies that \var{flist}
cannot be empty on line \ref{line:flist-pop}, so we do not run the risk
of dereferencing an invalid pointer on line \ref{line:buf-init}.
In the pseudo-code, we use \var{T*} to denote a pointer to an object of type $T$ and \var{Value[L]} to denote an array of $L$ word-sized values.
If \var{var} is a variable, \var{\&var} is used to denote the address of that variable.

\textbf{Initialization.} Each \wllsc{} object starts off pointing to a different
Buffer object and each free list is initialized with $2P$ distinct Buffers.
Buffers in the free lists are not pointed to by any of the \wllsc{} objects
and no Buffer appears in two free lists. This property is maintained as the 
algorithm executes.

\begin{figure*}
\begin{minipage}[t]{.46\textwidth}
	\begin{lstlisting}[linewidth=.99\columnwidth, numbers=left]
shared variables:
	Buffer* A[P]; // announcement array

local variables:
	list@$<$@Buffer*@$>$@ flist;
	list@$<$@Buffer*@$>$@ rlist;
	// initial size of flist is 2P
	// rlist is initially empty

struct Buffer {
	// Member Variables
	Value[L] val;
	int pid;
	bool seen;

	void init(Value[L] initialV) {
		copy initialV into val
		pid = -1; seen = 0; } };

struct WeakLLSC {
	// Member Variables
	Buffer* buf;
	
	// Constructor
	WeakLLSC(Value[L] initialV) {  
		buf = new Buffer(); 
		buf->init(initialV); }
\end{lstlisting}
\end{minipage}\hspace{.3in}
\begin{minipage}[t]{.48\textwidth}
\StartLineAt{28}
\begin{lstlisting}[linewidth=.99\textwidth, xleftmargin=5.0ex, numbers=left]
  optional<Value[L]> $\textbf{wLL}$() {
		Buffer* tmp = buf; @\label{line:wll-read}@
		A[pid].write(tmp); @\label{line:wll-announce}@
		if(buf == tmp)     @\label{line:wll-lin}@
			return tmp->val; @\label{line:wll-ret}@
		else return empty; }

  bool $\textbf{VL}$() {
  	Buffer* old = A[pid].read();
  	return buf == old; }@\label{line:wvl-lin}@

  bool $\textbf{SC}$(Value[L] newV) {
		Buffer* old = A[pid].read();
		Buffer* newBuf = flist.pop();    @\label{line:flist-pop}@
		newBuf->init(newV);           @\label{line:buf-init}@
		bool b = CAS(&buf, old, newBuf); @\label{line:wsc-lin}@
	  if(b) retire(old);
	  else flist.add(newBuf);          @\label{line:flist-push}@
		A[pid].write(NULL);
		return b; }

	void retire(Buffer* old) {
		rlist.add(old);                  @\label{line:rlist-add}@
		if(rlist.size() == 2*P) {
		  list@$<$@Buffer*@$>$@ reserved = [];
		  for(int i = 0; i < P; i++)
		    reserved.add(A[i].read());
		  newlyFreed = rlist \ reserved; @\label{line:setdiff}@
		  rlist.remove(newlyFreed);       @\label{line:rlist-remove}@
		  flist.add(newlyFreed); }}};    @\label{line:flist-add}@
	\end{lstlisting}
\end{minipage}
\caption{Amortized constant time implementation of $L$-word Weak LL/SC
  from CAS. Code for process with id \var{pid}.}
\label{alg:weakllsc}
\end{figure*}

\textbf{linear-time set difference.} The operation \var{rlist \textbackslash{} reserved} on
line \ref{line:setdiff} represent set difference.
What makes Hazard Pointers \emph{expected} rather than \emph{worst-case} 
constant time is that they use a hash table to perform these two steps. 
Instead, we add some space for meta-data in each \var{Buffer} object so that it can 
store a process id, \var{pid}, and a bit, \var{seen}. To perform the set 
difference \var{rlist \textbackslash{} reserved}, the process first visits 
each buffer \var{B} in \var{rlist} and prepares the buffer by setting \var{B.pid}
to its own process id and setting \var{B.seen} to false. Then, the process 
loops through \var{reserved} and for each buffer, it 
sets \var{seen} to true if \var{pid} equals its own process id. Next, 
the process loops through \var{rlist} again and constructs a list of 
buffers that have not been seen. This list is the result of the set 
intersection. Finally, 
the process has to reset everything by setting \var{B.pid} to $\bot$ 
for each \var{B} in \var{rlist}.

\textbf{Deamortization.} So far, the algorithm we have described takes \emph{amortized} constant time. 
To deamortize it, each process can maintain two sets of retired list and free 
lists. 
Each time the process pops from one free list, it performs a constant amount 
of work towards populating the other. 

\textbf{Space complexity.} The algorithm uses $P$ shared space for the 
announcement array, $O(P^2)$ local space for all the retired and free lists, and 
$O((M+P^2)L)$ shared space for all the buffers and \wllsc{} objects. 
Therefore, its total space usage is $O((M+P^2)L)$. 
In addition, it only uses pointer-width read, write, CAS as atomic operations, so it fulfills the claims in Result \ref{result:wllsc}.


\subsection{Correctness Proof}
\label{sec:wllsc-proof}



We begin by defining some useful terms and then reasoning about the lifecycle of a buffer. 
We will use $M$ to denote the number of \wllsc{} objects.
A buffer can be in one of the following $2P+M$ possible states: it can be pointed to
by a \wllsc{} object, it can be in the retired list of some process,
or it can be in the free list of some process. 
We consider a buffer to be in the retired list of a process if it is in that process's \var{rlist} or if no \wllsc{} object points to it and it is about to be added to that process's \var{rlist}. 
Similarly, we consider a buffer to be in the free list of a process if it is in that process's \var{flist} or if it has been popped off that process's \var{flist} and not yet written into any \wllsc{} object.
We can show by induction that a buffer cannot be in two different states at the same time. 
For example, if a buffer is in a process's free list, then it cannot be an any process's retired list and it cannot be pointed to by any \wllsc{} object. 
We will make use of this fact several times throughout our correctness proof.

The next step is to prove that the linear-time algorithm we described for set difference is correct.

\begin{lemma}
\label{lem:setdiff}
  The algorithm we described for linear-time set difference (Section \ref{sec:wllsc-alg}) is correct when used on line \ref{line:setdiff} of Figure \ref{alg:weakllsc}.
\end{lemma}

\begin{proof}
Recall that in the set difference algorithm, each buffer has an extra \var{seen} and \var{pid}
field, and that these fields are only accessed during the set difference computation. 

We begin by arguing that the set difference algorithm is correct as long as no process writes to the \var{pid} or \var{seen} field of a buffer that is in another process's retired list. 
Recall that the first step of the algorithm (when executed by $p_i$) is to set \var{pid} to $i$, and \var{seen} to \var{false} for each buffer in $p_i$'s retired list. 
Then for each buffer in \var{reserved}, it sets \var{seen} to \var{true} if \var{pid} equals $i$. 
The set of buffers in $p_i$'s retired list with \var{seen} equal to \var{false} are returned. 
Note that $p_i$'s retired list remains the same throughout this computation. 
If no other process writes to the \var{pid} or \var{seen} field of any buffer in $p_i$'s retired list, then this computation behaves as if it was executed in a sequential setting and so it returns the correct value.

All that remains is to prove that no process writes to the \var{pid} or \var{seen} field of a buffer that is in another process's retired list. 
Since no buffer can be in two different retired lists, it suffices to show that whenever $p_i$ writes to the \var{pid} or \var{seen} field of a buffer, that buffer is in $p_i$'s retired list. 
From the description of the algorithm, we can see that this holds for the \var{pid} fields. We focus on proving this for the \var{seen} fields. 
The only place where this could potentially not holds is when process $p_i$ loops through the buffers in \var{reserved} and for each buffer, sets \var{seen} to \var{true} if \var{pid} equals $i$.

We argue that a buffer's \var{pid} equals $i$ only if the buffer is in $p_i$'s retired list. 
The \var{pid} field of each buffer is initially $\bot$ and during a set difference operation by process $p_i$, the \var{pid} fields of all the buffers in $p_i$'s retired list get temporarily set to $i$ and then reset to $\bot$. 
Since $p_i$'s retired list stays the same throughout its set difference operation, all the \var{pid} fields that get set to $i$ are reset to $\bot$. 
Therefore whenever $p_i$ sees a buffer with \var{pid} equal to $i$, it knows the buffer is in its retired list.
From the description of the algorithm, we can see that $p_i$ only sets \var{seen} to \var{true} if the buffer is in its retired list.
\end{proof}




Next, we define the linearization points for \wll{}, \wvl{} and \wsc{} operations.

\begin{definition}
The linearization point of a \wsc{} operation is on line \ref{line:wsc-lin}. 
\wvl{} operations are linearized on line \ref{line:wvl-lin}.
For a \wll{} operation, its linearization point depends on whether or not it was successful. 
A successful one is linearized on line \ref{line:wll-lin} whereas an unsuccessful one is linearized at its first step.
\end{definition}

Let $E$ be an execution history of the \wllsc{} implementation. 
We assume $E$ is a valid execution history where $p_i$ invokes an \wsc{} or a \wvl{} on the object \var{X} only after a successful \wll{} on \var{X}.
We also assume that there is at most one outstanding \wll{} operation per process in $E$.
At each configuration, we define the \emph{value} of a \wllsc{} object \var{X}
to be the $L$-word value stored in \var{X.buf->val}.
We define 
To prove that a \wllsc{} object \var{X} is linearizable, it suffices to prove the following properties:

\begin{enumerate}
  \item The value of \var{X} only changes at the linearization point of a 
        successful \var{X.SC} operation.
  \item The linearization point of a successful \var{X.SC(newV)} operation changes 
        the value of \var{X} to \var{newV}.
  \item A successful \var{X.wLL} operation returns the value of \var{X} at its 
        linearization point.
  \item An \var{X.SC} operation $S$ by process $p$ is successful if and only if no successful 
        \var{X.SC} is linearized between the linearization points of 
        $S$ and the last successful \var{X.wLL} before $S$ by process $p$.
  \item An \var{X.VL} operation $V$ by process $p$ is successful if and only if no successful 
        \var{X.SC} is linearized between the linearization points of 
        $V$ and the last successful \var{X.wLL} before $V$ by process $p$.
  \item If \var{X.wLL} is unsuccessful then a successful 
        \var{X.SC} linearized during its execution interval.
\end{enumerate}

To help prove these properties, we make the following observations that are easy
to verify by examining the pseudo-code. The first observation is a weaker version of Property 1.

\begin{observation}
\label{obs:xbuf-change}
  The value of \var{X.buf} can only be changed at the linearization point of a 
  successful \var{X.}\wsc{} operation.
\end{observation}

The following observation states the converse and it holds because a process's free list
never contains a buffer that is being pointed to by \var{Y.buf} for any
\wllsc{} variable \var{Y}. This implies that \var{old != newbuf} at the
linearization point of each successful \wsc{} operation.

\begin{observation}
\label{obs:xbuf-change2}
Each successful \var{X.}\wsc{} operation changes \var{X.buf} at its linearization
point.
\end{observation}

\textbf{Proof of Property 1. } Let $s$ be a step in the execution history $E$.
Suppose $s$ is not the linearization point of a successful \var{X.}\wsc{} operation. 
We want to show that $s$ could not have changed the value of \var{X}.
By observation \ref{obs:xbuf-change}, we know that $s$ could not have changed the
value of \var{X.buf}. Therefore, we only need to worry about writes to the array \var{X.buf->val}. 
This array can only be written to if \var{X.buf} is in some process's
free list and it cannot be in any process's free list because it is being pointed to by \var{X.buf}.

\textbf{Proof of Property 2. } To prove this property, we just need to show that
\var{newbuf->val} equals \var{newV} on line \ref{line:wsc-lin} of a \wsc{}(\var{newV}) operation.
This holds because \var{newV} was written to \var{newbuf->val} on line \ref{line:buf-init} and no other process can write to \var{newbuf->val}
between the start of line \ref{line:buf-init} and the execution of line 
\ref{line:wsc-lin}. This is because a process can only write to buffers
that it has in its free list.

\textbf{Proof of Property 3. }
A successful \var{X.wLL} operation returns on line
\ref{line:wll-ret}. This line is not atomic because \var{tmp->val} could be an array
of words. Since \var{tmp->val} contains the value of \var{X} at line \ref{line:wll-lin} by definition, it suffices to show that this array
cannot be written to between line \ref{line:wll-lin} and the end of the \var{wLL} operation. This would mean that the \var{X.wLL} operation sees a consistent
snapshot of the array on line \ref{line:wll-ret}, and moreover, it would mean that
the $L$-word value that \var{X.wLL} reads from \var{tmp->val} is equal to the value of \var{X} at the \var{X.wLL}'s linearization point.
Thus, it suffices to show that \var{tmp->val} cannot be written to between line \ref{line:wll-lin} and the end of the \var{wLL} operation.
To show this, we take advantage of the fact that \var{X.buf} and \var{A[i]} both point to the
same buffer, \var{tmp}, at line \ref{line:wll-lin}. 
Also note that \var{A[i]} remains equal to \var{tmp} until the end of the \var{wLL} operation. As long as \var{A[i]} equals \var{tmp}, \var{tmp} cannot appear in the
free list of any process and we prove this fact in Claim \ref{clm:reserved}.
If \var{tmp} cannot appear in the free list of any process between line \ref{line:wll-lin} and the end of the \var{wLL} operation, then 
the array \var{tmp->val} cannot be written to in this interval.

\begin{claim}
\label{clm:reserved}
If at configuration $C$, both \var{X.buf} and \var{A[i]} point to buffer
$b$, then until \var{A[i]} changes, $b$ cannot appear in the free list of any process.
\end{claim}

\begin{proof}
At configuration $C$, $b$ cannot be in any process's retired or free lists
because it is being pointed to by \var{X.buf}.
In order for $b$ to appear in a free list, $b$ must first be
added to a process's retired list, then move onto that process's free list.
However, after $b$ is added to a 
process's retired list, that process will not add $b$ to its free list as
long as \var{A[i]} points to $b$.
\end{proof}

\textbf{Proof of Properties 4 and 5. } Both properties follow directly from the following claim.

\begin{claim}
\label{clm:prop4}
Let $O$ be either a \wsc{} or a \wvl{} operation and let $L$ be the last successful \wll{} operation before $O$ by the same process.
$O$ returns true if and only if no successful \wsc{} operation linearized between the linearization points of $L$ and $O$.
\end{claim}

\begin{proof}
Let $p_i$ be the process that performed $O$.
Since there is at most one outstanding \wll{} per process, we know that 
$p_i$ does not perform any \wll{} operations on any other \wllsc{} object between $L$ and $O$.
Therefore $A[i]$ does not change between the linearization points of $L$ and $O$.
This means that $O$ returns true if and only if the value of \var{X.buf} is the same at the linearization points of $L$ and $O$.
Thus, it suffices to show that the value of \var{X.buf} is the same
at the two linearization points if and only if no successful \wsc{} operation linearized between them. The backwards direction follows directly from Observation \ref{obs:xbuf-change}.

For the forwards direction, let $b$ represent the value \var{X.buf} at the linearization points of $L$ and $O$.
Since $L$ was a successful \wll{}, we know that at the 
linearization point of $L$, \var{X.buf} and \var{A[i]} both store the pointer $b$.
Since $A[i]$ does not change between the linearization points of $L$ and $O$, by Claim \ref{clm:reserved}, $b$ cannot appear in the free list of any process during this interval.

Suppose for contradiction that a 
successful \wsc{} operation $S'$ linearized in this interval. $S'$ must have
changed \var{X.buf} by observation \ref{obs:xbuf-change2}. Also by Observation 
\ref{obs:xbuf-change2}, another \var{X.SC} operation must have linearized
between the linearization points of $L$ and $O$ that changed 
\var{X.buf} back to $b$. This is a contradiction because $b$ cannot appear
in the free list of any processes during this interval.
\end{proof}

\textbf{Proof of Property 6. } In order for a \var{X.wLL} to be unsuccessful, the
value of \var{X.buf} must have changed between lines \ref{line:wll-read} and 
\ref{line:wll-lin}. By observation \ref{obs:xbuf-change}, there must have been a
successful \var{X.SC} that linearized in this interval, which completes the proof.

    \section{Single-Writer Atomic Copy}
\label{sec:swcopy}

The copy primitive, \swcopy{}, can be used to 
atomically read a value from some source memory location and write it 
into a \dest{} object. It is similar to the \memcopy{} primitive that 
was studied in~\cite{herlihy1991wait}, except that our \dest{} objects are
single-writer and we allow the source memory location
to be modified by any instruction (e.g. write, 
fetch-and-add, swap, CAS, etc). The sequential specifications of \swcopy{} and \dest{} objects are given below.


\begin{definition}
\label{def:swcopy}
A \dest{} object supports 3 operations \rd{}, \wt{} and \swcopy{} 
with the following sequential specifications:
  \begin{itemize}
      \item \rd{}\var{()}: returns the current value in the 
      \dest{} object 
      (initially $\bot$).
      \item \wt{}\var{(Value v)}: sets \var{v} as the current value 
      of the \dest{} object.
      \item \swcopy{}\var{(Value* addr)}: reads the value pointed to 
      by \var{addr} and sets it as the current value of the \dest{}
      object.
  \end{itemize}
Any number of processes can perform \rd{} operations, but only one 
process is allowed to \wt{} or \swcopy{} into a particular \dest{} 
object.
\end{definition}

We restricted this interface to be single-writer because it 
was sufficient for the
use cases we consider. We find that single-writer \dest{} objects 
are very useful in announcement array based algorithms where it
is beneficial for the read and the announcement to happen atomically.
It's possible to generalize this interface to support 
atomic copies that concurrently write to the same destination object. 
However, it is unclear what the desired behavior should be in 
this case.
One option would be to give
atomic copy `store' semantics where the value of the \dest{} object
is determined by the last \wt{} or \mwcopy{} to that location. 
Another option would be to give atomic copy `CAS' semantics where 
the \mwcopy{} is only successful if the \dest{} object stores the
expected value. 
The right choice of definition will likely depend on the potential application. Section \ref{sec:swcopy-alg} describes our implementation of 
\swcopy{}.

\subsection{Algorithm for Single-Writer Atomic Copy}
\label{sec:swcopy-alg}

In this section, we show how to implement \dest{} objects that
support \rd{}, \wt{}, and \swcopy{} in $O(1)$ time and 
$O(M+P^2)$ space (where $M$ is the number of \dest{} objects).
Our algorithm only requires pointer-width read, write and CAS instructions. 

We represent a \dest{} object \var{D} internally using a triplet, 
\var{D.val}, \var{D.ptr}, and \var{D.old}. When there is no \swcopy{} in 
progress, \var{D.val} stores the 
current value of the \dest{} object. When there is a copy in progress, \var{D.ptr} stores a 
pointer to the location that is being copied from. Operations
that see a copy in progress will help complete the copy. Finally, 
\var{D.old} stores the previous value of the \dest{} object. 
The variables \var{D.val} and \var{D.ptr} are stored together in a 
\wllsc{} object (defined in Section \ref{sec:wllsc}). This allows us 
to read from and write to them atomically as well as prevent any potential 
ABA problems. The downside is that the only way to read \var{D.val} 
or \var{D.ptr} is through a \wll{} operation which can repeatedly fail 
due to concurrent \wsc{} operations. For this reason, we keep 
\var{D.old} in a separate object, so that the readers can return 
\var{D.old} if they fail too many times on \wll{}. Readers will only perform \wsc{} operations that change \var{D.ptr} from not \var{NULL} to \var{NULL}. Therefore, the writer's \wll{} will be successful whenever \var{D.ptr} is \var{NULL}. 
We will maintain the invariant that \var{D.ptr} is \var{NULL} whenever there is no concurrent \swcopy{}.
We also ensure that \var{D.ptr} changes exactly twice during each \swcopy{}.
The first change writes a valid pointer and the second change resets it back to \var{NULL}.

A \swcopy{}\var{(Value* src)} on \dest{} object 
\var{D} begins by backing up the current value from \var{D.val} 
into \var{D.old}. 
At this point, \var{D.ptr} is guaranteed to be
\var{NULL}, so the writer can successfully read \var{D.val} with a
\wll{}. The \swcopy{} proceeds by writing \var{src} into \var{D.ptr}
with a \wsc{}.
Finally, it reads the value \var{v} pointed to by \var{src} and tries to write (\var{v}, \var{NULL}) into (\var{D.val}, \var{D.ptr}) with a \wsc{}.
It's not a problem if the \wsc{} 
fails because that means another process has helped complete the copy. 

To \rd{} from \var{D}, a process begins by trying to read the pair 
(\var{D.val}, \var{D.ptr}) with a \wll{}. If it fails on this 
\wll{} twice, then it is safe to return \var{D.old} because the 
value of $D$ has been updated at least once during this \rd{}. Now 
we focus on the case where one of the \wll{}s succeed and reads (\var{D.val}, \var{D.ptr}) into local variables (\var{val}, \var{ptr}). If \var{ptr} is NULL, then \var{val} 
stores the current value, which the \rd{} returns. If 
\var{ptr} is not NULL, then there is a concurrent \swcopy{} 
operation and the \rd{} tries to help by reading the value 
\var{v} referenced by \var{ptr} and writing (\var{v}, \var{NULL}) 
into (\var{D.val}, \var{D.ptr}) with a \wsc{}. If the \wsc{} is 
successful, then the \rd{} returns \var{v}. Otherwise, the process performs
one last \wll{}. If it is successful and sees that \var{D.ptr} is \var{NULL},
then it returns \var{D.val}. Otherwise, it is safe to return \var{D.old}.

The \wt{} operation is the most straightforward to implement. Since each \dest{} object only has a single writer, a 
\wt{} operation simply uses a \wll{} and a \wsc{}
to store the new value into \var{D.val}.
There cannot be any successful \wsc{} operations concurrent with the \wll{}
because the other processes can only succeed on a \wsc{} 
during a \swcopy{} operation. Therefore, the \wll{} and \wsc{} performed by the 
\wt{} will both always succeeds. The \wt{} operations also needs to keep \var{D.old} up to
data so it, updates it before performing the \wsc{}.

In our algorithm, we assumed that the source objects fit in a single word so that they can be atomically read from and written to. However, this assumption is not necessary. The algorithm can be 
generalized to work for larger source objects as long as they support an 
atomic \rd{} operation.

Pseudo-code is shown in Figure \ref{alg:swcopy}. 
From the pseudo-code, we can see that each operation takes constant time.
To implement $M$ \dest{} objects, it uses $M$ \wllsc{} objects, each spanning two words, and $O(M)$ pointer-width read, write, CAS objects. Using the algorithm from Result \ref{result:wllsc} to implement the \wllsc{} objects, we get an overall space usage of $O(M+P^2)$, which satisfies the properties in Result \ref{result:swcopy}.

\begin{figure*}
\begin{minipage}[t]{.46\textwidth}

  \begin{lstlisting}[linewidth=.99\columnwidth, numbers=left]
struct Data {Value val; Value* ptr;};
struct Destination {
  // Member Variables
  WeakLLSC<Data> data; 
  // data is initially <|@$\bot$@, NULL|>
  Value old;

  void $\textbf{swcopy}$(Value *src) {
    // This wLL() cannot fail
    old = data.wLL().val;          @\label{line:swcopy-ll1}@
    data.SC(<|empty, src|>);       @\label{line:swcopy-sc1}@
    Value val = *src;              @\label{line:swcopy-src}@
    optional<Data> d = data.wLL(); @\label{line:swcopy-ll2}@
    if(d.hasValue() && d.ptr != NULL) @\label{line:swcopy-if}@
      data.SC(<|val,NULL|>); }     @\label{line:swcopy-sc2}@
\end{lstlisting}
\end{minipage}\hspace{.3in}
\begin{minipage}[t]{.48\textwidth}
\StartLineAt{16}
\begin{lstlisting}[linewidth=.99\textwidth, xleftmargin=5.0ex, numbers=left]
void $\textbf{write}$(Value new_val) {
  // This wLL() cannot fail
  old = data.wLL().val;                @\label{line:write-ll}@
  data.SC(<|new_val, NULL|>); }        @\label{line:write-sc}@

Value $\textbf{read}$() {
  optional<Data> d = data.wLL();       @\label{line:read-wll1}@
  if(!d.hasValue()) {
    d = data.wLL();                    @\label{line:read-wll2}@
    if(!d.hasValue()) return old;}        @\label{line:read-old1}@
  if(d.ptr == NULL) return d.val;      @\label{line:read-val1}@
  value v = *(d.ptr);                  @\label{line:read-src}@
  if(data.SC(<|val, NULL|>)) return v; @\label{line:read-sc}@
  d = data.wLL();                      @\label{line:read-wll3}@
  if(d.hasValue() && d.ptr == NULL)    
    return d.val;                      @\label{line:read-val2}@
  return old; } };                     @\label{line:read-old2}@
  \end{lstlisting}
\end{minipage}
 
\caption{Atomic copy (single-writer). 
  Code for process with id \var{pid}.}
\label{alg:swcopy}
\end{figure*}

    \subsection{Correctness Proof}
\label{sec:swcopy-proof}

We begin by defining the linearization points of \wt{} and \swcopy{}. 
The linearization point of \rd{} is more complicated, so we will differ its definition until later.
Each \wt{} operation is linearized on line \ref{line:write-sc}.
For \swcopy{} operations, we will prove in Claim \ref{clm:lin-exists} that there exists exactly one \wsc{} instruction $S$ during the \swcopy{} that sets \var{data.ptr} to \var{NULL} and that this instruction either happens on line \ref{line:swcopy-sc2} of the \swcopy{} or line \ref{line:read-sc} of a concurrent \rd{} $R$. 
If $S$ from line \ref{line:swcopy-sc2}, then the \swcopy{} is linearized when it executes line \ref{line:swcopy-src}. 
Otherwise, the \swcopy{} is linearized on line \ref{line:read-src} of $R$.
We show in Claim \ref{clm:lin-contains} that this linearization point is contained in the execution interval of the \swcopy{}. 
Note that partially complete \swcopy{} operations without a \wsc{} instruction setting \var{data.ptr} to \var{NULL} are not linearized. 

For the purposes of this proof, we will focus on an execution $E$ consisting of operations on a single \dest{} object \var{D}.
For simplicity, we will write \var{data.ptr} instead of \var{D.data.ptr} and \swcopy{} instead of \var{D.}\swcopy{}.
At each configuration $C$ in $E$, we define the \emph{current value} of \var{D} to be the value written by the last modifying operation (either a \wt{} or a \swcopy{}) linearized before $C$. 
To show that the algorithm in Figure \ref{alg:swcopy} is correct, it suffices to show that the value returned by each \rd{} operation is the value of \var{D} at some step during the \rd{}.
The \rd{} is linearized at that step.

We first prove two useful claims about the structure of the algorithm.
Throughout the proof, it's important to remember that there can only be one \wt{} or \swcopy{} operation active at any time.
We say that a pointer is \emph{valid} if it is not \var{NULL}.

\begin{claim}
\label{clm:rd-null}
Suppose the \wsc{} performed by a \rd{} operation is successful, then \var{data.ptr} was valid at all configurations between line \ref{line:read-val1} of the \rd{} and the \wsc{}.
\end{claim}

\begin{proof}
Let $R$ be a \rd{} operation with a successful \wsc{} operation $S$ on line \ref{line:read-sc}.
Let $L$ be the successful \wll{} operation corresponding to $S$. 
$L$ was either executed on line \ref{line:read-wll1} of $R$ or line \ref{line:read-wll2} of $R$. 
Since $S$ is successful, \var{data.ptr} cannot have changed between $L$ and $S$. 
If \var{data.ptr} was \var{NULL} in this interval, then the if statement on line \ref{line:read-val1} would have evaluated to true, and $S$ would not have been executed. 
Therefore, \var{data.ptr} is valid at all configurations between $L$ and $S$, which includes all configurations between line \ref{line:read-val1} of $R$ and $S$.
\end{proof}

\begin{claim}
\label{clm:swcopy-null}
Suppose the \wsc{} on line \ref{line:swcopy-sc2} of a \swcopy{} operation is successful, then \var{data.ptr} is valid at all configurations between line \ref{line:swcopy-sc1} of the \swcopy{} and the \wsc{}.
\end{claim}

\begin{proof}
Let $Y$ be a \swcopy{} operation with a successful \wsc{} operation $S$ on line \ref{line:swcopy-sc2}.
For $S$ to be executed, the if statement on line \ref{line:swcopy-if} must evaluate to true, which means that the \wll{} operation $L$ on line \ref{line:swcopy-ll2} must have been successful.
Since $S$ is a successful \wsc{}, \var{data.ptr} cannot have changed between $L$ and $S$. 
Again, due to the if statement on line \ref{line:swcopy-if}, \var{data.ptr} is valid in this interval.

It remains to show that \var{data.ptr} is valid between lines \ref{line:swcopy-sc1} and \ref{line:swcopy-ll2}.
Suppose for contradiction that \var{data.ptr} is \var{NULL} in this interval.
The only operation that can change \var{data.ptr} to be valid is on line \ref{line:swcopy-sc1} of \swcopy{}, so \var{data.ptr} would have remained \var{NULL} until the end of $Y$. 
This contradicts the fact that \var{data.ptr} is valid between $L$ and $S$.
Therefore \var{data.ptr} is valid at all configurations between line \ref{line:swcopy-sc1} of $Y$ and $S$.
\end{proof}

The following two claims show that the linearization points of each \swcopy{} operation is well-defined and lie within its execution interval.

\begin{claim}
\label{clm:lin-exists} 
  There is exactly one successful \wsc{} instruction during a \swcopy{} $Y$ that sets \var{data.ptr} to \var{NULL} and this \wsc{} instruction is either from line \ref{line:swcopy-sc2} of $Y$ or line \ref{line:read-sc} of some \rd{}. Furthermore, this \wsc{} instruction is executed after the first \wsc{} of $Y$.
\end{claim}

\begin{proof}
Let $Y_i$ be the $i$th \swcopy{} operation in $E$. 
The order is well defined because there can be only one \swcopy{} operation active at a time. 
We proceed by induction on $i$, alternating between two different propositions. 
Let $P_i$ be the proposition that \var{data.ptr} equals \var{NULL} at the start of $Y_i$. 
Let $Q_i$ be the proposition that Claim \ref{clm:lin-exists} holds for $Y_i$. 
$P_1$ acts as our base case and for the inductive step, we show that $P_i$ implies $Q_i$ and that $Q_i$ implies $P_{i+1}$.

For the base case, we know that \var{data.ptr} is initialized to \var{NULL} and it can only be changed to something that is valid by the first \wsc{} of a \swcopy{} operation. 
Therefore, \var{data.ptr} remains \var{NULL} until the first \swcopy{} operation.

To show that $P_i$ implies $Q_i$, we use the same argument to argue that \var{data.ptr} is \var{NULL} between the first \wll{}/\wsc{} pair performed by $Y_i$. 
By Claim \ref{clm:rd-null}, no \wsc{} operation from a \rd{} can succeed between the first \wll{}/\wsc{} pair of $Y_i$.
This means the first \wsc{} performed by $Y_i$ (on line \ref{line:swcopy-sc1}) is guaranteed to succeed and set \var{data.ptr} to something valid. 
Between the first \wsc{} of $Y_i$ and the end of $Y_i$, the only two operations that could possibly change $Y_i$ are the \wsc{} on line \ref{line:swcopy-sc2} of $Y_i$ and line \ref{line:read-sc} of a \rd{} operation.
During this interval, if there are no successful \wsc{} operations from line \ref{line:read-sc}, then the \wsc{} on line \ref{line:swcopy-sc2} of $Y_i$ is guaranteed to execute and succeed.
This shows that there is at least one successful \wsc{} from line \ref{line:swcopy-sc2} or line \ref{line:read-sc} between the first \wsc{} and the end of $Y_i$.
By Claim \ref{clm:swcopy-null}, the \wsc{} on line \ref{line:swcopy-sc2} cannot succeed if \var{data.ptr} is \var{NULL}, and similarly for the \wsc{} on line \ref{line:read-sc} (Claim \ref{clm:rd-null}).
Since the \wsc{}s on lines \ref{line:swcopy-sc2} and \ref{line:read-sc} both set \var{data.ptr} to \var{NULL}, at most one such \wsc{} can succeed between the first \wsc{} of $Y_i$ and the end of $Y_i$. Therefore, $P_i$ implies $Q_i$.

All that remains is to show that $Q_i$ implies $P_{i+1}$. 
From $Q_i$, we know that \var{data.ptr} gets set to \var{NULL} between the first \wsc{} of $Y_i$ and the end of $Y_i$.
It will remain \var{NULL} until the first \wsc{} of $Y_{i+1}$, which means it is \var{NULL} at the beginning of $Y_{i+1}$.
\end{proof}

\begin{claim}
\label{clm:lin-contains}
  The linearization point of each \swcopy{} operation $Y$ lies between the first \wsc{} and the end of $Y$.
\end{claim}

\begin{proof}
A \swcopy{} operation $Y$ is either linearized at line \ref{line:swcopy-src} of its own operation or line \ref{line:read-src} of a \rd{} operation $R$. 
Clearly, this lemma holds in the former case, so we focus on the latter. 

By Lemma \ref{clm:lin-exists}, we know that the \wsc{} operation $S$ on line \ref{line:read-sc} of $R$ happens between the first \wsc{} of $Y$ and the end of $Y$. 
This means that the successful \wll{} operation $L$ corresponding to $S$ must have happened after the first \wsc{} of $Y$ and before $S$. 
From the code, we can see that line \ref{line:read-src} of $R$ (which is the linearization point of $Y$) happens between $L$ and $S$. 
By transitivity, the linearization point of $Y$ happens between the first \wsc{} of $Y$ and the end of $Y$.
\end{proof}

The next claim is useful for arguing that \var{data.ptr} is \var{NULL} at all configurations during a \wt{} operation and at all configurations between the beginning and the first \wsc{} of a \swcopy{}.

\begin{claim}
\label{clm:null-interval}
\var{data.ptr} can only be valid between the first \wsc{} of a \swcopy{} and the end of the \swcopy{}.
\end{claim}

\begin{proof}
\var{data.ptr} is initially \var{NULL} and the only instruction that sets \var{data.ptr} to something valid is the first \wsc{} of a \swcopy{} instruction.
By Claim \ref{clm:lin-exists}, we know that after this \wsc{} instruction and before the end of the \swcopy{}, \var{data.ptr} is set back to \var{NULL}.
Therefore, \var{data.ptr} can only be valid between the first \wsc{} of a \swcopy{} and the end of the \swcopy{}.
\end{proof}

Finally, we prove the main claim.

\begin{claim}
\label{clm:cur-val}
  If \var{data.ptr} is \var{NULL}, then \var{data.val} stores the current value of \var{D}.
\end{claim}

\begin{proof}
We will prove this by induction on the execution history $E$. 
The fields of \var{D} are initialized so that \var{data.ptr} stores \var{NULL} and \var{data.val} stores the initial value of \var{D}. 
Therefore this claim holds for the initial configuration. 
Suppose, for induction, that this claim holds for some configuration $C$, we need to show that it holds for the next configuration $C'$. 
If \var{D.data.ptr} is valid in $C'$, then the claim is vacuously true, so suppose \var{D.data.ptr} is \var{NULL} at $C'$. 
Let $S$ be the step between $C$ and $C'$. 
There are four cases for $S$; either (1) $S$ is a successful \wsc{} operation from line \ref{line:swcopy-sc2}, (2) $S$ is a successful \wsc{} operation from \ref{line:read-sc}, (3) $S$ is a successful \wsc{} operation from line \ref{line:write-sc}, or (4) $S$ is not a successful \wsc{} on \var{data}.

In the first case, $S$ is executed by a \swcopy{} operation $Y$, which is linearized on line \ref{line:swcopy-src} of $Y$. 
The value written into \var{data.val} by $S$ is equal to the value of the source location at the linearization point of $Y$. 
There cannot be any \swcopy{} or \wt{} operation linearized between the linearization point of $Y$ and $S$, so \var{data.val} stores the current value of \var{D} at $C'$.

For the second case, we first show that $S$ occurs during some a \swcopy{} operation. 
Due to the if statement on line \ref{line:read-val1}, $S$ can only be successful if \var{data.ptr} is valid. 
By Claim \ref{clm:null-interval}, \var{data.ptr} can only be valid during a \swcopy{} operation, which means that $S$ must occur during some \swcopy{} operation $Y$.
By Claim \ref{clm:lin-exists}, we know that $Y$ is linearized on line \ref{line:read-src} of the \rd{} operation that executed $S$.
Since there can only be one \swcopy{} or \wt{} at a time, there cannot be any other \swcopy{} or \wt{} operation linearized between the linearization point of $Y$ and $S$. 
Since the value written into \var{data.val} by $S$ is equal to the value of the source location at the linearization point of $Y$, \var{data.val} stores the current value of \var{D} at $C'$.

For case (3), $S$ is the linearization point of a \wt{} operation and $S$ writes the value of that \wt{} operation into \var{data.val}. 
This means \var{data.val} stores the current value of \var{D} at $C'$.

Finally, for the fourth case, suppose $S$ is not a successful \wsc{} on \var{data}.
This means the value of \var{data.val} will remain unchanged between $C$ and $C'$. 
By the inductive hypothesis, \var{data.val} stores the current value of \var{D} at $C$, so it suffices to show that
there are no \wt{} or \swcopy{} operation linearized at $S$.
By Claims \ref{clm:rd-null} and \ref{clm:swcopy-null}, \var{data.ptr} is valid at the linearization point of a \swcopy{} operation. 
Since \var{data.ptr} is \var{NULL} both before and after $S$, no \swcopy{} operation can be linearized at $S$.
To show that no \wt{} operations can be linearized at $S$, it suffices to show that the \wsc{} at the linearization point of a \wt{} operation is always successful.
Let $W$ be a \wt{} operation by process $p$.
The only \wsc{} operations on \var{data} that can be concurrent with $W$ are from \rd{} operations.
By Claim \ref{clm:null-interval}, \var{data.ptr} is \var{NULL} for the duration of $W$, and by Claim \ref{clm:rd-null}, no \wsc{} from a \rd{} operation can succeed during $W$. Therefore, both the \wll{} and the \wsc{} performed by $W$ are guaranteed to succeed.
\end{proof}


Suppose $R$ is a completed \rd{} operation that returns $v$.
As previously noted, to prove that Figure \ref{alg:swcopy} is a linearizable implementation of a \dest{} object, it suffices to show that there exists a step during $R$ such that the value of the \dest{} object at that step is equal to $v$. We linearize $R$ at that step. If there are multiple operations linearized at the same step, \rd{} operations are always linearized last. Note that there cannot be multiple \wt{} or \swcopy{} operations linearized at the same step.

There are five possible return points for a \rd{} operation. 
If $R$ returns on lines \ref{line:read-sc} or \ref{line:read-val2}, then on lines \ref{line:read-sc} or \ref{line:read-wll3} (respectively), we know that \var{data.val} equals $v$ and \var{data.ptr} equals \var{NULL}.
If $R$ returns on line \ref{line:read-val1}, then either on line \ref{line:read-wll1} or line \ref{line:read-wll2}, \var{data.val} equals $v$ and \var{data.ptr} equals \var{NULL}.
By Claim \ref{clm:cur-val}, \var{data.val} stores the current value of the \dest{} object whenever \var{data.ptr} is \var{NULL}, so for these three return points there exists a step during $R$ such that $v$ is the current value.

Now suppose $R$ returns on lines \ref{line:read-old1} or \ref{line:read-old2} (i.e. the case where $R$ reads and returns the value in \var{D.old}).
There must have been two successful \wsc{}s, $S_1$ and $S_2$, on \var{D.data} during $R$. 
In the case where $R$ returns on line \ref{line:read-old1}, these two successful \wsc{} operations occurred during the \wll{}s on lines \ref{line:read-wll1} and \ref{line:read-wll2}. 
In the case where $R$ returns on line \ref{line:read-old2}, $S_1$ was the one that caused the \wsc{} on line \ref{line:read-sc} to fail and $S_2$ occurred during the \wll{} on line \ref{line:read-wll3}.
By Claims \ref{clm:rd-null} and \ref{clm:swcopy-null}, there cannot be two successful \wsc{}s from lines \ref{line:swcopy-sc2} or \ref{line:read-sc} in a row without a successful \wsc{} from line \ref{line:swcopy-sc1} of \swcopy{} or line \ref{line:write-sc} of \wt{} in between.
Therefore, there must have been a successful \wsc{} either from line \ref{line:swcopy-sc1} of \swcopy{} or line \ref{line:write-sc} of \wt{} during $R$. We'll use $S$ to denote this \wsc{} operation. 
In both cases, the line immediately before $S$ updates \var{D.old} by first performing a \wll{} on \var{data}. 
By Claim \ref{clm:null-interval}, \var{data.ptr} equals \var{NULL} during this \wll{} operation and since the only \wsc{} operations that could potentially cause it to fail are by \rd{} operations, by Claim \ref{clm:rd-null}, this \wll{} is guaranteed to succeed. 
By Claim \ref{clm:cur-val}, \var{data.val} stores the current value $v'$ at the time of this \wll{} operation. 
This value gets written into \var{old}, so \var{old} stores the current value immediately after this step. Since there is only a single \wt{} or \swcopy{} at a time, \var{old} still contains the current value immediately before $S$. 
$R$ reads and returns the value of \var{old} at its last step so there are two cases. 
Either $R$ reads $v'$ from \var{old} or it reads something newer. 
If $R$ reads $v'$, then it returns the current value of \var{D} at the step immediately before $S$ (which happens during $R$). 
If $R$ reads something newer, then \var{old} must have been updated between $S$ and the end of $R$. 
This can only happen on line \ref{line:swcopy-ll1} or on line \ref{line:write-ll}, and we've already argued that \var{old} stores the current value of \var{D} on these two lines. 
Therefore, in either case, $R$ returns a value that was the current value of \var{D} at some point during $R$.

\section{LL/SC from CAS}
\label{sec:llsc}

Now we have all the tools we need to implement LL/SC from CAS (Result \ref{result:llsc}). 
We begin, in Section \ref{sec:llsc-alg}, by presenting an algorithm that works whenever there is at most one outstanding LL per process.
Then in Section \ref{sec:fixed-k}, we show how to generalize this to support $k$ outstanding LLs per process.

\subsection{Implementation of LL/SC from CAS}
\label{sec:llsc-alg}

This algorithm is almost identical to our algorithm for weak LL/SC from CAS 
(Section \ref{sec:wllsc-alg}). To ensure that the \LL{} operation always succeeds, 
we use \swcopy{} to atomically read and announce the current buffer 
(lines \ref{line:wll-read} and \ref{line:wll-announce} of Figure \ref{alg:weakllsc}). 
This means that the announcement array needs to be an array 
of \dest{} objects (from Section \ref{sec:swcopy-alg}) rather than raw pointers.
Other than that, the algorithm remains the same.
Figure \ref{alg:llsc} shows the difference between this algorithm and the weak 
LL/SC algorithm from Figure \ref{alg:weakllsc}.

  \begin{figure*}

  \begin{lstlisting}[linewidth=.99\columnwidth, numbers=left]
Destination<Buffer*> A[P]; 

struct LLSC {
  Buffer* buf;
  ...
  Value[L] $\textbf{LL}_i$() {
    A[pid].swcopy(&buf);          @\label{line:ll-swcopy}@
    Buffer* tmp = A[pid].read();
    return tmp->val; } @\label{line:ll-ret}@
};
  \end{lstlisting}
  \caption{Amortized constant time implementation of $L$-word LL/SC from CAS. The algorithm is exactly the same as Algorithm \ref{alg:weakllsc} except for the parts that are shown. Code for process with id \var{pid}.
  Note that the type of the announcement array has changed, so way we \rd{} from and \wt{} to the announcement array is also different.}
  \label{alg:llsc}
  \end{figure*}

This algorithm uses $O((M+P^2)L)$ pointer-width read, write, CAS objects just like 
in Figure \ref{alg:weakllsc}, but it also uses $P$ \dest{} objects for the announcement 
array. From Result \ref{result:swcopy}, we know that $P$ \dest{} objects can be 
implemented in constant time and $O(P^2)$ space, so this algorithm achieves the 
bounds in Result \ref{result:llsc}.



\subsection{LL/SC Correctness Proof (outline)}
\label{sec:llsc-proof}

In the proof of correctness for our \wllsc{} algorithm, the key property 
we made use of is that at the linearization point of a successful \var{D.}\wll{} operation,
\var{A[pid]} and \var{D.buf} both point to the same buffer. 
We linearize our \LL{} operation from Figure \ref{alg:llsc} so that the same property holds.
At the linearization point of the \var{swcopy} operation on 
line \ref{line:ll-swcopy}, both \var{A[pid]} and \var{D.buf} store the same value, so we linearize \LL{} operations at this point.
\SC{} and \VL{} operations are linearized just as they were in the \wllsc{} algorithm.
This way, we can basically reuse the proof from Section \ref{sec:wllsc-proof}. 
The first two properties hold without modification because \SC{} operations are exactly the same in both algorithms. 
Property 6 is unnecessary because \LL{} operations are always treated as successful. 
For Properties 3, 4 and 5, and Claims \ref{clm:reserved} and \ref{clm:prop4}, we just need to replace line numbers from the \wll{} pseudo-code (in Figure \ref{alg:weakllsc}) with line numbers from the \LL{} pseudo-code (in Figure \ref{alg:llsc}). 
For example line \ref{line:wll-ret} in Figure \ref{alg:weakllsc} would be replaced with line \ref{line:ll-ret} in Figure \ref{alg:llsc}.

\subsection{Handling multiple outstanding LL operations per process}
\label{sec:fixed-k}

\begin{figure}[!htbp]
\begin{minipage}[t]{.42\textwidth}

  \begin{lstlisting}[linewidth=.99\columnwidth, numbers=left]
shared variables:
  Destination<Buffer*> A[P]@\hl{[k]}@; 
  // announcement array

local variables:
  @\hl{list<int> freeSlots;}@
  @\hl{//initialized to \{0,1,...,k-1\}}@
  list@$<$@Buffer*@$>$@ flist;
  list@$<$@Buffer*@$>$@ rlist;
  // initial size of flist: @\hl{2kP}@
  // rlist is initially empty

struct Buffer {
  // Member Variables
  Value[L] val;
  int pid;
  bool seen;

  void init(Value[L] initialV) {
    copy initialV into val
    pid = -1; seen = 0; } };

struct LLSC {
  // Member Variables
  Buffer* buf;
  
  // Constructor
  LLSC(Value[L] initialV) {  
    buf = new Buffer(); 
    buf->init(initialV); }
\end{lstlisting}
\end{minipage}\hspace{.3in}
\begin{minipage}[t]{.51\textwidth}
\StartLineAt{28}
\begin{lstlisting}[linewidth=.99\textwidth, xleftmargin=5.0ex, numbers=left]
  <Value[L], @\hl{int}@> $\textbf{LL}$() {
    @\hl{int slot = freeSlots.pop();}@
    A[pid][slot].swcopy(&buf);           @\label{line:llk-lin}@
    Buffer* tmp = A[pid][slot].read();
    return <tmp->val, @\hl{slot}@>; }

  void @\hl{\textbf{CL}}@(int slot) {
    A[pid][slot].write(NULL);
    freeSlots.push(slot); }

  bool $\textbf{VL}$(@\hl{int slot}@) {
    Buffer* old = A[pid][slot].read();
    return buf == old; } 

  bool $\textbf{SC}$(Value[L] newV, @\hl{int slot}@) {
    Buffer* old = A[pid][slot].read();
    Buffer* newBuf = flist.pop();    
    newBuf->init(newV);          
    bool b = CAS(&buf, old, newBuf);     @\label{line:sck-lin}@
    if(b) retire(old);
    else flist.add(newBuf);          
    A[pid][slot].write(NULL);
    @\hl{freeSlots.push(slot);}@
    return b; }

  void retire(Buffer* old) {
    rlist.add(old);                  
    if(rlist.size() == @\hl{2kP}@) {    @\label{line:scan-if}@
      list@$<$@Buffer*@$>$@ reserved = [];
      for(int i = 0; i < P; i++)
        @\hl{for(int j = 0; j < k; j++)}@
          reserved.add(A[i][j].read());
      newlyFreed = rlist \ reserved; 
      rlist.remove(newlyFreed)       
      flist.add(newlyFreed); }}};    
  \end{lstlisting}
\end{minipage}
\caption{Amortized constant time implementation of $L$-word LL/SC from CAS supporting $k$ outstanding LLs by a single process. Code for process with id \var{pid}. The difference between this and Figures \ref{alg:weakllsc} + \ref{alg:llsc} are highlighted.}
\label{alg:llsc-fixed-k}
\end{figure}

So far, we've worked under the assumption that each process only has at most a single outstanding LL operation at any time. 
In this section, we slightly modify the algorithm from Section \ref{sec:llsc-alg} so that it can handle up to $k$ outstanding \LL{} operations per process.
The new algorithm has the same time complexity and $\Theta((M+kP^2)L)$ space complexity.

Before describing the algorithm, we first have to modify the \LL{}/\SC{} interface slightly to efficiently support large $k$.
In addition to returning a value, \LL{} now also has to return a handle.
This handle is passed as an argument to future \SC{}, and \VL{} operations.
In a valid execution, we require that whenever process $p$ performs a \var{X}.\SC{}(\var{newV}, \var{h}) or a \var{X}.\VL{}(\var{h}) operation, \var{h} must have been the handle returned by the previous \var{X}.\LL{} operation by process $p$.
The same interface modifications were done in~\cite{moir1997practical}.

Pseudo-code for the new algorithm is shown in Figure \ref{alg:llsc-fixed-k} with all the major modifications highlighted.
The main change to give each process $k$ announcement slots rather than a single one.
To keep track of which ones are free, each process maintains a local list in a variable called \var{freeSlots}, which is initialized to $\{0, ..., k-1\}$.
Process $p_i$ begins an \LL{} operation by popping an index $i$ off the \var{freeSlots} list and using that index to determine which of its announcement slots to use.
The \LL{} operation then proceeds as before and it returns $i$ as the handle.
An \SC{}(\var{newV}, \var{slot}) operation uses \var{slot} to determine which of its announcement slots was used by the corresponding \LL{}.
Before returning, the \SC{} operation adds \var{slot} back to \var{freeSlots} so that it can be used again in future \LL{} operations.

We also support a \CL{}(\var{slot}) (Clear-Link) operation which simply adds \var{slot} back to \var{freeSlots} and clears the corresponding location in the announcement array. 
This is used by the programmer to indicate that he or she no longer wishes to perform an \SC{} on that LL/SC variable, which reduces the number of outstanding LL operations.

A process $p$ will never run out of announcement slots because the number of slots that are "in-use" (i.e. not in \var{freeSlots}) is exactly the number of outstanding LL operations by $p$.

In the \retire{} operation, we now have to scan $kP$ announcements, so to make the amortization argument work out, we will make sure this scan happens at most once every $kP$ calls to \retire{}.
This is achieved by initializing the free list (\var{flist}) with $2kP$ buffers and performing a full scan only if the retired list (\var{rlist}) reaches $2kP$ elements (line \ref{line:scan-if}).

All operations are still linearized at the same configurations as before. For example \LL{} operations are still linearized on line \ref{line:llk-lin} and \SC{} operations are still linearized on line \ref{line:sck-lin}.
To see why these modifications preserve the correctness of the original LL/SC algorithm, consider an \LL{}/\SC{} pair by process $p_i$.
Suppose $h$ is the index returned by the \LL{} operation. The key property $A[i][h]$ does not change between the linearization points of \LL{} and \SC{}, just like in the original algorithm.

After applying the same deamortization and set difference algorithm as in Section \ref{sec:wllsc-alg}, we get $\Theta(L)$ time for \LL{} and \SC{}, $O(1)$ time for \VL{} and \CL{}, while using $\Theta((M+kP^2)L)$ space.
This algorithm satisfies all the properties from Result \ref{result:llsc}.

\section{Conclusion and Discussion}
\label{sec:conclusion}

We introduced a new primitive called \var{swcopy} and shown how to implement it efficiently.
We used this primitive to implement constant time LL/SC from CAS in a way that is both space efficient and avoids the use of unbounded sequence numbers.
We believe the \var{swcopy} primitive can simplify the design of many other concurrent algorithms and make reasoning about them more modular.

Our LL/SC from CAS algorithm assumes that $k$, the maximum number of outstanding LLs per process, and $P$ are known in advance.
It would be interesting to see if the same time and space bounds are possible without requiring advanced knowledge of $k$ or $P$, and without unbounded sequence numbers.
One way to support dynamically changing $k$ and $P$ is to implement the announcement array in our algorithms as a linked list~\cite{michael2004hazard}.
However, adding and removing from the linked list would take more than constant time.
	
	\bibliographystyle{abbrv}
	\bibliography{../../biblio} 
\end{document}